\documentclass[10pt]{article}
\usepackage[utf8]{inputenc}
\usepackage{mathrsfs}  
\usepackage{titlesec}

\titleformat*{\section}{\large\bfseries}
\usepackage{setspace}
\usepackage{amssymb,amsmath,amsthm,amsfonts,amstext}
\usepackage[english]{babel}
\usepackage[left=1cm,right=1cm,top=1cm,bottom=1.5cm]{geometry}
\usepackage{enumerate}
\usepackage[mathscr]{euscript}

\usepackage{mathtools,braket}
\usepackage{bm, hyperref}
\usepackage{xcolor, float}
\usepackage[capitalize]{cleveref}

\newtheorem{theorem}{Theorem}

\newtheorem{proposition}[theorem]{Proposition}
\newtheorem{lemma}[theorem]{Lemma}

\newtheorem{claim}[theorem]{Claim}

\newtheorem{definition}{Definition}

\newcommand{\del}{\ensuremath{\delta}}
\newcommand{\din}{\ensuremath{\delta^{\mathrm{in}}}}

\newcommand{\opt}{\textsc{OPT}}
\newcommand{\safe}{\mathscr{S}}
\newcommand{\unsafe}{\mathscr{U}}

\newenvironment{proofof}[1]{\begin{proof}[{Proof of #1}]}{\end{proof}}

\title{\Large \textbf{A \boldmath$2$-Approximation Algorithm for Flexible Graph Connectivity}}
\author{\large
Sylvia Boyd\thanks{{\tt sboyd@uottawa.ca}. School of Electrical Engineering and Computer Science, University of Ottawa.}
\and
Joseph Cheriyan\thanks{ {\tt \{jcheriyan,sharat.ibrahimpur\}@uwaterloo.ca}. Department of Combinatorics and Optimization, University of Waterloo.}
\and
Arash Haddadan\thanks{{\tt ahaddada@virginia.edu}. Biocomplexity Institute and Initiative, University of Virginia.}
\and 
\addtocounter{footnote}{-2}
Sharat Ibrahimpur\footnotemark
}
\date{}

\begin{document}
\maketitle
\vspace{-15pt}
\abstract{We present a 2-approximation algorithm for the Flexible Graph Connectivity problem \cite{AHM20} via a reduction to the minimum cost $r$-out $2$-arborescence problem.}

\vspace{-5pt}
\section{Introduction}
In this paper, we consider the Flexible Graph Connectivity (FGC) problem which was introduced by Adjiashvili, Hommelsheim and M\"uhlenthaler \cite{AHM20}.
In an instance of FGC, we have an undirected connected graph $G = (V,E)$, a partition of $E$ into unsafe edges $\unsafe$ and safe edges $\safe$, and nonnegative costs $\{ c_e \}_{e \in E}$ on the edges.
The graph $G$ may have multiedges, but no self-loops.
A subset $F \subseteq E$ of edges is feasible for FGC if for any unsafe edge $e \in F \cap \unsafe$, the subgraph $(V,F \setminus \{e\})$ is connected.
We seek a (feasible) solution $F$ minimizing $c(F) = \sum_{e \in F} c_e$.
The motivation for studying FGC is two-fold.
First, FGC generalizes many well-studied survivable network design problems.
Most notably, the minimum-cost $2$-edge connected spanning subgraph (2ECSS) problem corresponds to an instance of FGC where all edges are unsafe. 
Second, FGC captures a non-uniform model of survivable network design problems where a subset of edges never fail, i.e., they are always safe.
Adjiashvili et al.~\cite{AHM20} gave a $2.523$-approximation algorithm for FGC.
Our main contribution is a simple $2$-approximation algorithm for FGC.
At a high level, our result is based on a straightforward extension of the $2$-approximation algorithm of Khuller and Vishkin \cite{khuller-vishkin} for 2ECSS.

\begin{theorem} \label{fgc-2apx}
There is a $2$-approximation algorithm for FGC.
\end{theorem}

Adjiashvili et al. \cite{AHM20} also consider the following  generalization of FGC.
Let $k \geq 1$ be an integer.
A subset $F \subseteq E$ of edges is feasible for the $k$-FGC problem if for any edge-set $X \subseteq F \cap \unsafe$ with $|X| \leq k$, the subgraph $(V,F \setminus X)$ is connected.
The goal in $k$-FGC is to find a solution of minimum cost.
The usual FGC corresponds to $1$-FGC.
The following result generalizes Theorem~\ref{fgc-2apx}.


\begin{theorem} \label{kfgc-apx}
There is a $(k+1)$-approximation algorithm for $k$-FGC.
\end{theorem}

Our proof of Theorem~\ref{kfgc-apx} is based on a reduction from $k$-FGC to the minimum-cost $(k+1)$-arborescence problem (see \cite{schrijver-book}, Chapters~52 and~53).
We lose a factor of $k+1$ in this reduction.
Fix some $k$-FGC solution $F$ and designate a vertex $r \in V$ as the root vertex.
For an edge $e = uv$, we call the arc-set $\{(u,v),(v,u)\}$ as a bidirected pair arising from $e$.
The key idea in our proof is that there exists an arc-set $T$ that contains $k+1$ arc-disjoint $r \rightarrow v$ dipaths for each $v \in V \setminus \{r\}$ while satisfying the following two conditions: (i) for an unsafe edge $e = uv \in F$, $T$ uses at most $2$ arcs from a bidirected pair arising from $e$; and (ii) for a safe edge $e = uv \in F$, $T$ uses at most $k+1$ arcs from the disjoint union of $k+1$ bidirected pairs arising from $e$. 
This argument is formalized in Lemma~\ref{fgc-to-arb}.
Complementing this step, we show that any arc-set $T$ (consisting of appropriate orientations of edges in $E$) that contains $k+1$ arc-disjoint $r \rightarrow v$ dipaths for every $v \in V \setminus \{r\}$ can be mapped to a $k$-FGC solution.



\vspace{-12pt}
\section{A \boldmath $(k+1)$-Approximation Algorithm for $k$-FGC}
\vspace{-5pt}
For a subset of vertices $S$ and a subgraph $H$ of $G$, we use $\del_H(S)$ to denote the set of edges in $H$ that have one endpoint in $S$ and the other in $V \setminus S$. 
The following characterization of $k$-FGC solutions is straightforward.

\begin{proposition} \label{cut-prop}
$F$ is feasible for $k$-FGC $\iff \forall \, \emptyset \subsetneq S \subsetneq V$, $\del_F(S)$ contains a safe edge or $k+1$ unsafe edges.
\end{proposition}


For the rest of the paper, we assume that the given instance of $k$-FGC is feasible: this can be easily checked by computing a (global) minimum-cut in $G$ where we assign a capacity of $k+1$ to safe edges and a capacity of $1$ to unsafe edges.
Let $D = (W,A)$ be a digraph and $\{c'_a\}_{a \in A}$ be nonnegative costs on the arcs.
We remark that $D$ may have parallel arcs but it has no self-loops. 
Let $r \in W$ be a designated root vertex. For a subgraph $H$ of $D$ and a set of vertices $S\subseteq{W}$, we use $\din_H(S)$ to denote the set of arcs such that the head of the arc is in $S$ and the tail of the arc is in $W\setminus S$.
 
\begin{definition}[$r$-out arborescence]
An $r$-out arborescence $(W,T)$ is a subgraph of $D$ satisfying: (i) the undirected version of $T$ is acyclic; and (ii) for every $v \in W \setminus \{r\}$, there is an $r \rightarrow v$ dipath in $(W,T)$.
\end{definition}

\begin{definition}[$r$-out $k$-arborescence]
For a positive integer $k$, a subgraph $(W,T)$ is an $r$-out $k$-arborescence if $T$ can be partitioned into $k$ arc-disjoint $r$-out arborescences.
\end{definition}

\begin{theorem}[\cite{schrijver-book}, Chapter~53.8]\label{thm:k-arb-char}
Let $D=(W,A)$ be a digraph and let $k$ be a positive integer. For $r \in W$, the digraph $D$ contains an $r$-out $k$-arborescence if and only if $|\din_D(S)|\geq k$ for every nonempty $S \subseteq V\setminus \{r\}$.
\end{theorem}

\begin{claim} \label{lim-freq}
Let $(W,T)$ be an $r$-out $k$-arborescence for an integer $k \geq 1$.
Let $u,v \in W$ be any two vertices.
Then, the number of arcs in $T$ that have one endpoint at $u$ and the other endpoint at $v$ (counting multiplicities) is $\leq k$.
\end{claim}
\begin{proof}
Since an $r$-out $k$-arborescence is a union of $k$ arc-disjoint $r$-out $1$-arborescences, it suffices to prove the result for $k=1$.
The claim holds for $k=1$ because the undirected version of $T$ is acyclic, by definition.
\end{proof}

\begin{theorem}[\cite{schrijver-book}, Theorem~53.10] \label{min-cost-arb}
In polynomial time, we can obtain an optimal solution to the minimum $c'$-cost $r$-out $k$-arborescence problem on $D$, or conclude that there is no $r$-out $k$-arborescence in $D$.
\end{theorem}

The following lemma shows how a $k$-FGC solution $F$ can be used to obtain an $r$-out $(k+1)$-arborescence (in an appropriate digraph) of cost at most $(k+1)c(F)$.

\begin{lemma} \label{fgc-to-arb}
Let $F$ be a $k$-FGC solution.
Consider the digraph $D = (V,A)$ where the arc-set $A$ is defined as follows: for each unsafe edge $e \in F \cap \unsafe$, we include a bidirected pair of arcs arising from $e$, and for each safe edge $e \in F \cap \safe$, we include $k+1$ bidirected pairs arising from $e$.
Consider the natural extension of the cost vector $c$ to $D$ where the cost of an arc $(u,v) \in A$ is equal to the cost of the edge that gives rise to it.
Then, there is an $r$-out $(k+1)$-arborescence in $D$ with cost at most $(k+1)c(F)$.
\end{lemma}
\begin{proof}
Let $(V,T)$ be a minimum-cost $r$-out $(k+1)$-arborescence in $D$.
First, we argue that $T$ is well-defined. 
By Theorem~\ref{thm:k-arb-char}, it suffices to show that for any nonempty $S \subseteq V \setminus \{r\}$, we have $|\din_D(S)| \geq k+1$.
Fix some nonempty $S \subseteq V \setminus \{r\}$.
By feasibility of $F$, $\del_F(S)$ contains a safe edge or $k+1$ unsafe edges (see Proposition~\ref{cut-prop}).
If $\del_F(S)$ contains a safe edge $e = uv$ with $v \in S$, then by our choice of $A$, $\din_D(S)$ contains $k+1$ $(u,v)$-arcs.
Otherwise, $\del_F(S)$ contains $k+1$ unsafe edges, and for each such unsafe edge $uv$ with $v \in S$, $\din_D(S)$ contains the arc $(u,v)$.
Since $|\din_D(S)| \geq k+1$ in both cases, $T$ is well-defined. 

Finally, we use Claim~\ref{lim-freq} to show that $T$ satisfies the required cost-bound.
For each unsafe edge $e \in F$, $T$ contains at most $2$ arcs from the bidirected pair arising from $e$, and for each safe edge $e \in F$, $T$ contains at most $k+1$ arcs from the (disjoint) union of $k+1$ bidirected pairs arising from $e$.
Thus, $c(T) \leq 2c(F \cap \unsafe) + (k+1) c(F \cap \safe) \leq (k+1) c(F)$.
\end{proof}


Lemma~\ref{fgc-to-arb} naturally suggests a strategy for Theorem~\ref{kfgc-apx} via minimum-cost $(k+1)$-arborescences.

\begin{proofof}{Theorem~\ref{kfgc-apx}}
Fix some vertex $r \in V$ as the root vertex.
Consider the digraph $D = (V,A)$ obtained from our FGC instance as follows: for each unsafe edge $e \in \unsafe$, we include a bidirected pair arising from $e$, and for each safe edge $e \in \safe$, we include $k+1$ bidirected pairs arising from $e$.
For each edge $e \in E$, let $R(e)$ denote the multi-set of all arcs in $D$ that arise from $e$.
For any edge $e = uv \in E$ and arc $(u,v) \in R(e)$, we define $c_{(u,v)} := c_e$.
Let $(V,T)$ denote a minimum $c$-cost $r$-out $(k+1)$-arborescence in $D$.
By Lemma~\ref{fgc-to-arb}, $c(T) \leq (k+1) \opt$, where $\opt$ denotes the optimal value for the given instance of $k$-FGC.

We finish the proof by arguing that $T$ induces a $k$-FGC solution $F$ with cost at most $c(T)$.
Let $F := \{ e \in E : R(e) \cap T \neq \emptyset \}$.
By definition of $F$ and our choice of arc-costs in $D$, we have $c(F) \leq c(T)$.
It remains to show that $F$ is feasible for $k$-FGC.
Consider a nonempty set $S \subseteq V \setminus \{r\}$.
Since $T$ is an $r$-out $(k+1)$-arborescence, Theorem~\ref{thm:k-arb-char} gives $|\din_T(S)| \geq k+1$.
If $\din_T(S)$ contains a safe arc (i.e., an arc that arises from a safe edge), then that safe edge belongs to $\del_F(S)$.
Otherwise, $\din_T(S)$ contains some $k+1$ unsafe arcs (that arise from unsafe edges).
Since both orientations of an edge cannot appear in $\din_D(S)$, we get that $|\del_F(S) \cap \unsafe| \geq k+1$.
Thus, $F$ is a feasible solution for the given instance of $k$-FGC, and $c(F)\leq(k+1)\opt$.
\end{proofof}

{\small

}
\end{document}